\documentclass[journal]{IEEEtran}
\usepackage{pgfplots}
\usepackage{amsthm}
\usepackage{amsmath}
\usepackage{amssymb}
\usepackage{caption}
\usepackage{subcaption}
\usepackage{caption}
\usepackage[square, numbers, comma, sort&compress]{natbib}
\usepackage{algorithm}
\usepackage{algpseudocode}

\newtheorem{definition}{Definition}
\newtheorem{assumption}{Assumption}
\newtheorem{proposition}{Proposition}
\newtheorem{theorem}{Theorem}
\newtheorem{example}{Example}
\newtheorem{lemma}{Lemma}

\newtheorem{remark}{Remark}

\begin{document}
\title{A General Definition of Network Communities and the Corresponding Detection Algorithm}
\author{
\IEEEauthorblockN{Haoye Lu\\}
\IEEEauthorblockA{Department of Mathematics and Statistics, University  of Ottawa}

\IEEEauthorblockN{Amiya Nayak\\}
\IEEEauthorblockA{
SITE, University  of Ottawa,\\
Email: \{hlu044, nayak, sbehe070\}@uottawa.ca}}

\author{Haoye~Lu,~\IEEEmembership{Member,~IEEE,}
        Amiya~Nayak,~\IEEEmembership{Senior Member,~IEEE}
\thanks{Haoye~Lu and Amiya~Nayak are with School of Electrical Engineering and Computer Science, University of Ottawa,
Ontario, K1N 6N5 Canada e-mail: (hlu044@uottawa.ca; nayak@uottawa.ca).}
}

\maketitle{}

\begin{abstract}
Network structures, consisting of nodes and edges, have applications in almost all subjects. A set of nodes is called a community if the nodes have strong interrelations. Industries (including cell phone carriers and online social media companies) need community structures to allocate network resources and provide proper and accurate services. However, all the current detection algorithms are motivated by the practical problems, whose applicabilities in other fields are open to question. Thence, for a new community problem, researchers need to derive algorithms ad hoc, which is arduous and even unnecessary. In this paper, we represent a general procedure to find community structures in practice. We mainly focus on two typical types of networks: transmission networks and similarity networks. We reduce them to a unified graph model, based on which we propose a general method to define and detect communities. Readers can specialize our general algorithm to accommodate their problems. In the end, we also give a demonstration to show how the algorithm works.

\end{abstract}

\section{Introduction}
Our real world consists of elements associated by relations. We call the entity network made by elements with the relations among them. Most real world networks are not random as they reveal big inhomogeneity, high level of order and organization \citep{communityDetectioninGraphs}. For example, people working in a company may have much closer relations than the ones outside the company. The observation inspires people to partition the elements into groups (communities) such that the relations are strong and dense within the groups but sparse and weak among them \citep{6ac6eb110668494abd0c1fa99dc5daac, Fortunato201075,Newman12communities}.

Community detections have widespread applications. Amazon groups customers buying similar products together for better commodities recommendations. Facebook clusters the users by relationships,  hobbies, etc. to accurately suggest users with new friends and circles. Carriers group the locations among which customers have high transportation demands for a proper assignment of vehicles.

Because of the omnipresent community structures in practice, researchers try to find proper algorithms to detect them. There are mainly four traditional methods \citep{communityDetectioninGraphs}: graph partitioning \citep{kernighanLin}, hierarchical clustering \citep{HastieTrevorTEoS}, partitional clustering \citep{NeuralNetworksandComputationalIntelligence,Rattigan2007} and spectral clustering \citep{doi:10.1137/0603056,Luxburg:2007:TSC:1288822.1288832,Li2018b,Lu2018a}. These methods are designed for different purposes and reveal many fundamental properties of networks. After that, many related algorithms are proposed (for instances, modularity-based methods \citep{Fastalgorithmfordetectingcommunitystructureinnetworks,Chen2014, Kaur2017}, dynamic algorithms \citep{PhysRevLett.93.218701}, methods based on statistical inference \citep{Wan2016,Bickel15122009}, maximum likelihood \citep{Airoldi2008,7893802} and network motifs \citep{Milo824,lu2014,Benson163}). Recent research shows that modularity optimization and maximum likelihood algorithms are equivalent \citep{PhysRevE.94.052315}. 

Most algorithms work well in the areas from which they are derived. But the reliabilities outside their zones are controversial. The poor adaptabilities of the algorithms reveal the demand for general community detection methods as well as the general community definitions. Besides, most algorithms start from graphs directly, while the procedure of the math model abstraction is rarely formalized. In this paper, we try to tackle these problems. In particular,
\begin{enumerate}
	\item Is there a general way to abstract a concrete problem into a unified math model?
	\item Based on the unified math model, are there some common properties shared by most community definitions?
	\item Is there a common method to detect community structures?
\end{enumerate}
The rest of the paper is organized as follows. In Section~\ref{secRW}, we have a review on the popular community detection algorithms. And we also discuss their pros and cons. In Section~\ref{secCPR}, we introduce how to reduce a concrete problem into a graph model. Based on this, we define the community structure in Section~\ref{secComm}. In Section~\ref{secProp}, we propose and prove some propositions regarding our community definition. Then we provide a corresponding detection algorithm. A demonstration is given in Section~\ref{secDemo} and finally, we talk about the limitations of our model as well as the future work in Section~\ref{secLimit}.
\section{Related Work}\label{secRW}
The research of community detection starts from solving some concrete problems. For instance, Kernighan-Lin algorithm \citep{kernighanLin} is designed for clustering digital components into equal or nearly equal size communities such that mutual connections among the sets are minimized (due to the cost and stability consideration, electronic engineers need to minimize the number of connections among boards).

The researches always reduce the real-world network structures to the nodes connected by edges, where the edges represent the relations among the nodes. Although the meanings of relations vary in different papers, there are two main types: two nodes is related if
\begin{enumerate}
	\item there are material transmission between them, AND/OR
	\item they share some identical or similar properties.
\end{enumerate}

The material means concrete objects (like goods) or information (like data packages). One of the most typical examples is the transportation among cities. The cities easily communicating with each other are grouped into one community \citep{Guimera31052005}.

There are also many network structures constructed by the similarity of the properties of the nodes. In protein-protein interaction networks, biologists would cluster proteins with equivalent or similar functions into one group \citep{Chen:2006:DFM:1181979.1182343}. Then, the relations are the function similarities of the proteins. In social networks, the people active in the similar locations and/or the time slots could be considered as a community. Then the relations represent the location and schedule similarity. In World Wide Web, the communities may correspond to the groups of pages concerned with the related topics or events \citep{Dourisboure:2007:ECD:1242572.1242635, Chaabani2017}. Then the relations become the content similarity. 

Based on the graph model, lots of community detection algorithms have been proposed. 

Graph partitioning method groups the vertices into a predefined number of communities and minimize the number of edges among the groups. Most of the algorithms belonging to this method can perfectly solve particular problems in practice. However, the algorithms are not adapted to community detections due to the necessity of the pre-specified number of groups, which is in general unknown in community detection problems \citep{communityDetectioninGraphs}. More seriously, the graph partitioning method is not derived from an explicit definition of communities. So there is no guarantee that the vertex groups found by the method are communities following our intuitions.

Real-world networks commonly have hierarchical structures from which the abstracted graph models usually inherit. The corresponding detection algorithms fall into two types: agglomerative (bottom-up) approach and divisive (top-down) approach \citep{HastieTrevorTEoS}. Briefly, agglomerative approach starts from considering each node as a community and merges the community pairs as moving up the hierarchy. The divisible algorithm works in the opposite way. It starts by grouping all the nodes in one cluster and performs splitting recursively as moving down the hierarchy.

Partitional clustering also plays an important role in the graph clustering family. In order to apply the algorithm, the user must specify the number of clusters, which causes the same disadvantages as the ones belonging to graph partitioning \citep{communityDetectioninGraphs}. The method puts all nodes in a metric space, and thus the distance between each pair of the nodes in the space is defined. The distance used in the algorithm can be considered as a measure of dissimilarity between the nodes. The algorithm needs to cluster the nodes into a pre-specified number of groups which minimize a given cost function.

All community detection methods and techniques related to matrix eigenvectors belong to spectral clustering. The clustering method requires a distance function to measure the similarity among the objects. The fundamental idea behind the algorithm is to utilize the eigenvectors to cluster objects by connectedness rather than the distance. Although there are many successful applications in image segmentation and machine learning, researchers have already found several fundamental limitations. For instance, spectral clustering would fail if it uses the first $k$ eigenvectors to find $k$ clusters when confronting with clusters of various scales regarding a multi-scale landscape potential \citep{6287457}. Besides, the algorithm needs the assistance of partitional clustering, whose drawbacks are also inherited to spectral clustering.

Most of the aforementioned algorithms have remarkable performances in the problems they derived from. And their complexities are optimized for the dynamic network community detection \citep{PhysRevLett.93.218701,PhysRevLett.96.114102,PhysRevLett.101.168701}. However, researchers hardly focus on the general definition of communities and the corresponding general community detection algorithms. A few preliminary works have been done. In particular, the comparisons have been made among the different community definitions as well as the corresponding detection algorithms \citep{1742-5468-2005-09-P09008}. Besides, some researcher believes that the definition often depends on the specific system at hand and/or application one has in mind~\citep{communityDetectioninGraphs}.

\section{Concrete Problem Reduction}\label{secCPR}
Most people believe that, a community is some set of objects where the interrelations are strong. However, there are lots of arguments on the definitions of relations as well as the ways to measure them. Generally, we derive the relations in two ways. 
	
One is based on the \textbf{materials} transmission. The materials here can represent both material substance (like goods) or just information (like data). Intuitively, the objects that can easily communicate with each other should have strong relations among them. Then those objects can somehow be considered as a community. 
	
The other one is based on the similarity (for example, people buy the similar kind of books might be considered as a community). In usual cases, we believe that the objects in the same community should share some other similar properties. Then we can make some reasonable predictions on the community level (for example, Amazon uses this trick to recommend commodities).
	
Although two community derivations come from different motivations, we show that they can be reduced to the same graph model.

\subsection{Transmission network}
\subsubsection{Transmission relation characteristics}\label{TRC}
	In order to make the problem easier to discuss, only one material will be considered. Moreover, we need the following definitions and assumptions.
	
	\begin{assumption}
	For any material, there is a minimal unit can be transferred.  And we call the minimal unit a point.
	\end{assumption}
	
	\begin{definition}[Node]
	The objects that receives and sends points is nodes.
	\end{definition}
	
	\begin{definition}[Medium]
	The object that propagates points is medium. 
	\end{definition}
	
	\begin{assumption}\label{property implication}
		The transmission relations are constructed in nodes, media and points, which also determine the properties of the relations.
	\end{assumption}
	
	Because of Assumption~\ref{property implication}, the entity under the consideration consists of nodes, points and media. And we name it \textbf{transmission network}.
	
	Two most important characteristics in a transmission network are \textit{the number of points transferred} and \textit{the time consumed} in a transmission process. Their ratio is termed \textbf{speed}.
	\begin{definition}[Speed]
		Speed is the number of points transferred in a unit time interval.
	\end{definition}
	
	The transmission capability can be described by the speed function of time. For simplicity, only the node pairs connected by media directly are considered. By Assumption~\ref{property implication}, the behaviour of the speed function $f$ depends on the properties of the nodes, media and points. This means a specific analytic expression of the speed function cannot be given unless all the properties have been designated. However, some common characteristics can be expected. Suppose there is a pair of nodes $(u,v)$ connected by media directly. Then (see Figure~\ref{transmission rough graph}) the value of the speed function $f_{(u,v)}(t)$ remains zero until some point because of the latency caused by the sending, propagation and reception of the points. 

	\begin{figure}[h]
		\centering
		\includegraphics[width=0.8\columnwidth]{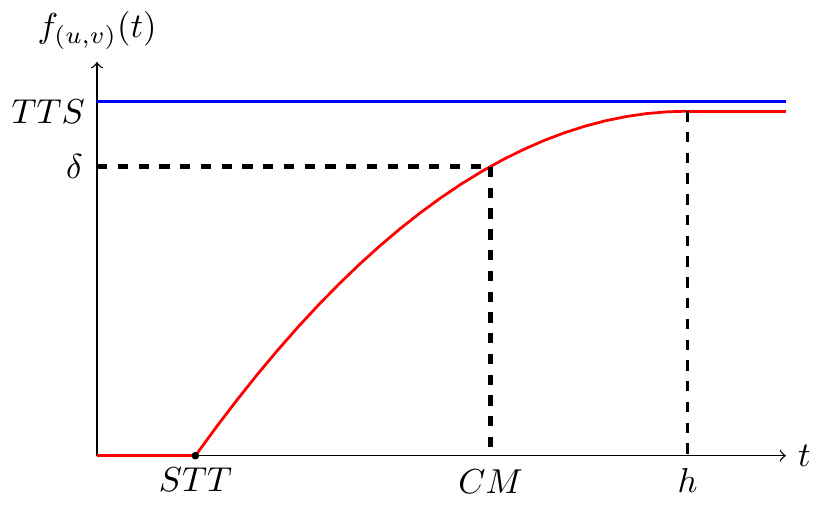} 
		\caption{A typical speed function function of time associated to node pair $(u,v)$.}\label{transmission rough graph}
	\end{figure}
		
	We call the moment that the function value becomes nonzero the \textbf{shortest transmission time (STT)} (theoretically, $STT = \text{sending time}+\text{transmission time}+\text{receiving time}$.). Besides, there is always an upper bound of the point propagation speed. And we call the least upper bound \textbf{theoretically transmission speed (TTS)}.
			
	However, in the real world, TTS may never be reached, which leads to the meaninglessness of it. So instead, we can choose some proper \textbf{threshold} $\delta \in \mathbb{R^*}$, the lowest acceptable speed. Correspondingly, we call the time to reach the threshold the \textbf{critical moment (CM)}.
	
	\begin{remark}
		Suppose there are two towns $A$ and $B$ near a river. $A$ is upstream of $B$. Consider the goods transportation on the river. Because of the water stream, $f_{(A,B)}(t) \geq f_{(B,A)}(t)$.	Thus, in general, $f_{(u,v)}(t) \neq f_{(v,u)}(t)$. 
	\end{remark}
	
	\subsubsection{Special relation strength measurement for transmission network}\label{SRSM}
		The speed functions of time can describe almost all the important properties of the relations in a transmission network. So, the special relation strength measurement (SRSM) can be derived from the function according to the problem confronting us. Here are several examples.
		\begin{enumerate}
			\item The shorter STT is, the stronger the corresponding relation is.
			\item The shorter CM is, the stronger the corresponding relation is.
			\item The shorter the time to transfer a certain number of points is, the stronger the corresponding relation is.
		\end{enumerate}
		
		Although a reasonable SRSM cannot be constructed until a concrete problem is given, several key properties should be shared by all SRSMs.
		
		Specifically, in a transmission network, relations should not cancel out each other. So the SRSM is non-negative. Besides, the strongest relation in a transmission network should be the one that the node relates to itself as the transmission speed can be considered as infinity. Intuitively, all SRSMs should have the same value in this case, which is assigned zero. Moreover, the relation strength will get weaker if the function value increases. Therefore, we have the following definition. 
		
		\begin{definition}[SRSM]
			Given some network, let $N$ be the set of nodes in it and $P \subseteq N \times N$ the set of node pairs connected by media directly. Then a function $s: P \rightarrow \mathbb{R^*}$ is a special relation strength measurement (SRSM) if it satisfies
			\begin{enumerate}
				\item $s(u,v) \geq 0$.
				\item $s(u,v) = 0$ if and only if $u = v$.
			\end{enumerate}
		\end{definition}
		
		\begin{remark}
			Since transmission functions have various values if we change the order of the parameters in general, so do SRSMs. Namely, $s(u,v) \neq s(v,u)$ in general.
		\end{remark}
		
		With the help of SRSM, the transmission network can be reduced to a \textbf{graph model}. In the graph, the nodes are represented by the vertices and connected by a weighted edge if they are connected by media directly. Moreover, the weights of the edges are assigned by a specific SRSM depending on the practical problem. In general, the graph should be directed. If $s(u,v) = s(v,u)$ for all applicable pairs of nodes $(u,v)$ in a network, the graph can also be considered as an undirected one.
		
		In Section~\ref{RSMFTN}, the graph model will be used to define more general relation strength measurement. Besides, \textit{node} and \textit{vertex} are used alternatively without considering the difference. So are \textit{relation} and \textit{edge}.  
	
	\subsubsection{Relation strength measurement for transmission network}\label{RSMFTN}
	In this part, the discussion is based on the graph model. Firstly, some notations of graphs need to be introduced.
	
	Let $G := (V_G, E_G)$ be a weighted graph where $V_G$ and $E_G$ are the sets of vertices and edges of graph $G$ respectively. For any edge $e \in E_G$, its weight is denoted by $|e|$. If there is no ambiguity about the choice of the graph, $V_G$ and $E_G$ are abbreviated to  $V$ and $E$. 
	
	In the previous section, SRSM measures the direct relations between any pair of nodes. That is, the retransmission function of nodes is not taken into account. Neither is the parallel transmission on various paths. However, in the real world, the retransmission and parallel transmission of points happen frequently (like express service and data transmission on the internet). So we have to derive a more general relation strength measurement function, which is named \textbf{relation strength measurement for transmission network (RSMFTN)}.
	
	For the same reason in the derivation of SRSM, an analytic expression of RSMFTN cannot be given until a concrete problem is designated. But a reasonable RSMFTN must hold several key properties.
	
	First of all, all the relations between any pair of nodes should be measurable. That is, the domain of RSMFTN is $V \times V$. The properties of SRSM should be inherited. Then RSMFTN is also non-negative and $RSMFTN(u,v) = 0$ if and only if $u = v$. Besides, there is no relation between the nodes belonging to two disconnected components. In contrast to the coincidence of two nodes, this is the other extreme. So it is reasonable to set the value of $RSMFTN$ to be infinity ($\infty$), an element greater than any real number. The relation of some vertices $u$ and $v$ gets stronger if the function value $RSMFTN(u,v)$ approaches to zero. 
	
	Consider a linear graph (Figure~\ref{linear graph}). The points sent by $u$ and received by $v$ must be retransmitted by $w$. So the difficulty to transfer points from $u$ to $v$ is not less than the one from $u$ to $v$ visiting $w$.	Notice that the difficulty of nodes to receive and send points has been included by SRSM and so indicated by the weights of the edges. Hence, the equality should hold. In other words, $RSMFTN(u,v) = RSMFTN(u,w)+RSMFTN(w,v)$ if $w$ is a cutting node on the path from $u$ to $v$. To add on, the transmission difficulty does not increase if we add some other retransmission node $s$ (see Figure~\ref{one more node}). Thence, in general, we have $RSMFTN(u,v) \leq RSMFTN(u,w)+RSMFTN(w,v)$. 
		\begin{figure}[h]
		\centering
			\begin{subfigure}[b]{0.4\columnwidth}
			\includegraphics[width=\columnwidth]{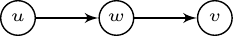}\caption{}\label{linear graph}
			\end{subfigure}\hspace{1em}
			\begin{subfigure}[b]{0.4\columnwidth}
			\includegraphics[width=\columnwidth]{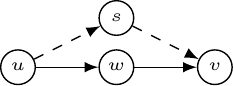}\caption{}\label{one more node}
		\end{subfigure}
		\caption{}
		\end{figure}
	
	Moreover, the ratio of the relation strengths (\textbf{the relative relation strength}) of two pairs of nodes should be fully determined by the relative magnitude of the weights. Also, for an indirect graph, the directions of edges are not taken into account; therefore, $\forall (u,v)\in V\times V. \hspace{0.5em} RSMFTN(u,v) = RSMFTN(v,u)$. 
	
	To sum up, we can define the RSMFTN as follows,
		
	\begin{definition}[Relation strength measurement for transmission network]
		 Suppose $G$ is some directed graph and $g: {V \times V} \rightarrow \mathbb{R} \cup \{\infty\}$ a function. Then $g$ is a relation strength measurement for transmission network (RSMFTN) if and only if for $u,v,w \in V$,		  
		 \begin{itemize}
		 	\item[1.] $g(u,v) \geq 0$. (non-negativity)
		 	\item[2.] $g(u,v)=0$ if and only if $u$ and $v$ coincide. (coincidence axiom)
		 	\item[3.] $g(u,v)= \infty$ if and only if there is no path between $u$ and $v$.
		 	\item[4.] $g(u,v) \leq g(u,w)+g(w,v)$. Moreover, the equality holds if the two components that contain $u$ and $v$ are connected by the cutting node $w$.
		 	\item[5.] Suppose $G'$ is a graph which is same as $G$ except that the edges' weights in $G'$ are all $\alpha$ times greater than the ones in $G$. Then for the corresponding vertices $u'$ and $v'$ in $G'$, $g(u',v') = \alpha g(u,v)$ 
		 \end{itemize}
		 Besides, if $G$ is undirected,
		 \begin{itemize}
		 	\item[6.] $g(u,v)=g(v,u)$ (symmetry)
		 \end{itemize} 
	\end{definition}	
	
	\begin{example}\label{SDFISRSM}
		Actually, many measurements derived by other researchers are RSMFTNs. A well-known one should be the shortest distance function (SDF), which evaluates the shortest distance between some pair of nodes in a graph. Here is the proof.
			\begin{proof}		
			Suppose $u$, $v$ and $w$ are arbitrary vertices in some directed graph $G$. Let $g$ denote the SDF. We prove the proposition when $G$ is weighted. The proof for the unweighted graph follows by setting the weights of the edges to one. 
			
			Since $g(u,v)$ returns the sum of weights on the shortest path from $u$ to $v$, $g(u,v) \geq 0$. The property $1$ holds. The properties $2$ and $3$ hold by the definition of SDF.
			
			For property 4, assume $g(u,v) > g(u,w)+g(w,v)$. Consider the path $p$ consisting of the shortest path from $u$ to $w$ and the one from $w$ to $v$. It is easy to see that the length of the path $l(p)$ is $g(u,w)+g(w,v)$, which is shorter than $g(u,v)$. So $g(u,v)$ cannot be the length of the shortest path. We get a contradiction. Moreover, if there exists a cutting node $w$ connecting the components  containing $u$ and $v$ respectively, the shortest path between $u$ and $v$ can be split into the one from $u$ to $w$ and the one from $w$ to $v$. So $g(u,v) = g(u,w)+g(w,v)$. Therefore, the property 4 holds.
			
			For property 5, suppose path $p$ is some shortest path from $u$ to $v$ in graph $G$. Then $l(u,v) = g(u,v)$. 
			Let $p'$ denote its counterpart in $G'$. Since the weights in $G'$ are $\alpha$ times the ones in $G$, so is the length of path $p'$. That is, $l(p') = \alpha l(p)$. Since path $p'$ connects $u'$ and $v'$ in $G'$, we have $g(u',v') \leq l(p')$. In other words, $g(u',v') \leq \alpha l(u,v)$. Similarly, consider the reverse transformation from $G'$ to $G$. That is, all the weights of edges in $G'$ is $\frac{1}{\alpha}$ times in the ones in $G$. So we have $\frac{1}{\alpha} l(u',v') \geq g(u,v)$, which is equivalent to $l(u',v') \geq \alpha g(u,v)$. Combining with $g(u',v') \leq \alpha l(u,v)$, we conclude that $g(u',v') =\alpha l(u,v)= \alpha g(u,v)$. So property 5 holds.
			
			Suppose $G$ is an undirected graph, then by the commutativity and associativity of the addition operator, $g(u,v) = g(v,u)$. In other words, the order to add the weights of the edges compounding the shortest path does not change the final result. 
			
			To sum up, SDF is an RSMFTN. 
			\end{proof}
	\end{example}


\subsection{Similarity Network}
	\subsubsection{Similarity function}
	In order to make the problem easier to discuss, we need to give some fundamental definitions at first. We name the objects that have similarity relations \textbf{nodes}. Each node may have various \textbf{properties}. Moreover, there are some possible options for each property (for example, red, blue, yellow are possible options for property colour), and we name these options \textbf{cases}.
	
	The similarity network concerns the property similarity of nodes. We assume that, for a certain problem, the set of properties is fixed and for each property, there exists a similarity function that maps the case pairs to real numbers. Intuitively, the a measure of similarity should not be negative. Thus, we assume that the similarity function is non-negative. Besides, to keep the consistency with the definition of SRSM given in Section~\ref{SRSM}, the function value increases while the similarity decreases. Moreover, for some objects $A$ and $B$, if $A$ is similar to $B$, then $B$ is also similar to $A$. Then we formalize the preliminary idea:

	\begin{definition}
		Let $P$ be some property of the nodes and $C_P$ the set of possible cases of $P$, then we say $s$ is a \textbf{similarity function} $s: P \times P \rightarrow \mathbb{R}$ if for $(\gamma_1,\gamma_2) \in C_P \times C_P$
		\begin{enumerate}
			\item $s(\gamma_1,\gamma_2) \geq 0$ (non-negativity)
			\item $s(\gamma_1,\gamma_2) = 0$ if and only if $p_1 = p_2$ (coincidence axiom)
			\item $s(\gamma_1,\gamma_2) = s(\gamma_2,\gamma_1)$ (symmetry)
		\end{enumerate}
	\end{definition}

	Since more than one properties $P_1, P_2, \cdots$ might be considered in general, we need to define a list of similarity functions $s_{P_1}, s_{P_2}, \cdots$. For convenience, we write them in matrix forms. That is, $[P_1 \ P_2 \cdots]$ and $[s_{P_1} \ s_{P_2} \cdots]$.
	  
	Two lists of similarities are not measurable. And the properties may not be of the identical importance. Thence, we need a function to translate a list of parameters to an index. Traditionally, we call the function manipulating the importances of a list factors \textbf{weight function}. So we have the following definition:
	\begin{definition}[Weight function]
		Suppose $N$ is the set of nodes under the consideration in some problem. Let $P_1, \ P_2 \ \cdots$ denote the properties. Besides, $[s_{P_1} \ s_{P_2} \cdots]$ is the list of the corresponding similarity functions.	A function of functions $w$ mapping $[s_{P_1} \ s_{P_2} \cdots]$ to a non-negative function $f$ is called a weight function.
	\end{definition}
	
	\begin{remark}\label{weightfuntionremark}
	The choice of the weight function depends on the practical problem we try to solve. A trivial weight function is just a list of weights. In more details, suppose $[s_{P_1} \ s_{P_2} \cdots s_{P_n}]$ is a list of similarity functions. Let $\alpha_1, \alpha_2 \cdots, \alpha_n$ be the weights indicating the importances. Then $w = [\alpha_1 \ \alpha_2 \cdots \ \alpha_n]$ can be a potential weight function. And $f$ is $$[\alpha_1 \ \alpha_2 \cdots \ \alpha_n] \times [s_{P_1} \ s_{P_2} \cdots s_{P_n}]^T = \sum^n_{i = 1} \alpha_i \cdot s_{P_i},$$ which is a non-negative function.
	\end{remark}
	
	\begin{example}
		Suppose we take two properties $P_1$ and $P_2$ into account. Besides, we have the possible cases $\gamma_1, \gamma_2, \gamma_3$ for $P_1$ and $\zeta_1, \zeta_2$ for $P_2$. For two nodes $u$ and $v$, assume $u$ has properties $\gamma_1$ and $\zeta_2$, and $v$ has properties $\gamma_3$ and $\zeta_1$. Let $s_{P_1}$ and $s_{P_2}$ be the similarity function we created for $P_1$ and $P_2$. If we use the way in remark \ref{weightfuntionremark} to define $w$, we have $w = [\alpha_1 \hspace{0.5em} \alpha_2]$. So we have $f(s,t) = \alpha_1 s_{P_1}(\mbox{s's $P_1$},\mbox{t's $P_1$})+\alpha_2 s_{P_2}(\mbox{s's $P_2$},\mbox{t's $P_2$})$. In particular, for nodes $u$ and $v$, we have $f(u,v) = \alpha_1\cdot s_{P_1}(\gamma_1, \gamma_3)+\alpha_2\cdot s_{P_2}(\zeta_2,\zeta_1)$. 
	\end{example}

	\subsubsection{Relation strength measurement for similarity network}
	
	The weight function can generate a function $f$ to measure the similarity of a pair of nodes. However, the weight function here has no guarantee that $f$ always gives the measurement following our intuition. In particular, we require $f$ satisfies the following properties:
	
	Suppose $N$ is the set of nodes under the consideration and $P$ the set of properties. Then for $u$, $v$ and $w$ in $N$, we have, 
	\begin{enumerate}
		\item $f(u,v) \geq 0$
		\item $f(u,v) = 0$ if and only if $u$ and $v$ have the exactly same cases for all properties in $P$
		\item $f(u,v) \leq f(u,w) + f(w,v)$
		\item $f(u,v) = f(v,u)$
	\end{enumerate}
	
	The first two properties are inherited from the similarity function. Since the similarity relation should be symmetric (that is, if $A$ is similar to $B$, then $B$ is also similar to $A$), so we have property $4$. Besides, property $3$ shows that the direct measurement of any pair of nodes is at least not greater than the sum of the ones with an intermediate point. Since this function is defined for similarity measurement, we name it the \textbf{relation strength measurement for similarity network (RSMFSN)}.
	
	\begin{remark}
		In other words, $f$ is a distance function. In fact, the example we give in remark \ref{weightfuntionremark} is an RSMFSN.
	\end{remark}

	\subsection{Relations between similarity network and transmission network}
	In many cases, there are very strong relations between similarity networks and transmission networks. A typical example is the pathogen infection among some species. If we consider the DNA similarity among organisms. It is easier for some certain pathogen to infect organisms that have similar DNAs. Or in other words, the easiness of pathogen transmission has a positive relationship with the DNA similarity. Therefore, the relative relation strength among the organisms should be similar whichever relation type we consider here.
	
	Since both RSMFTN and RSMFSN are used to measure the relations among the nodes in networks, and the follow-up propositions are based on their shared properties, we call both two relation measurements \textbf{relation strength measurement (RSM)} in the sequel. 
	
\section{Communities}\label{secComm}
	
		After defining RSM, the definition of communities can be derived. Ahead of giving a formal definition, an important problem needs to be discussed. That is, the community relation's transitivity. In other words, if $A$ and $B$ are contained in one community and so are $B$ and $C$, can we also say $A$ and $C$ are in one community? In general, this implication is not true. A typical counterexample is ``your friend's friends may not be your friends". So all relation strength between any pair of nodes should be considered when we define a community. Moreover, since only the groups of nodes having mutually strong enough relations are considered as communities, a relation strength threshold (\textbf{community parameter}) needs to be designated.
		
		\begin{definition}[Community]\label{defnCom}
			Suppose $W^G \subseteq V_G$ for some directed graph $G$, $\epsilon \in \mathbb{R^*}$, and $g$ is some RSM. Then $W^G$ is a community with respect to RSM $g$ and constant $\epsilon$ if and only if for all $(u,v) \in W^G \times W^G$. $ g(u,v) \leq \epsilon$. Moreover, we say $\epsilon$ is the community parameter (CP) of $W^G$ with respect to $g$. If there is no ambiguity of the choice of RSM and CP, we will briefly say $W^G$ is a community.
		\end{definition}

		Since CP gives a threshold of the relation strength, whichever pair of nodes we choose in a community, the relation strength of the pair cannot be weaker than the ones the CP represents. So for those problems considering the worst cases, the CP can be designated according to some CM with some certain threshold (in Section~\ref{TRC}). Then the inner structure of the community can be ignored since the poorest performance of the community satisfies the requirement. In other words, a community can be considered as a relatively independent entity, and the CP is a global property of it. 
	
		\begin{remark}
		Notice that the definition is based on the set of vertices instead of the subgraph used in many other papers. Besides, it is worthy to emphasize that the choice of communities usually consider the whole graph's topology rather than the local one (this shows that the community is some higher level structure based on the original graph). Since the results might be different for various choices of graph topology, the superscripts are used to make the description clear (for example, $W^G$ means $W$ is a vertex set and graph $G$ is the working topology).
		\end{remark}
			
		\section{Propositions and detection algorithm for communities}\label{secProp}
			Based on the definition of RSM, an adjoint complete graph can be derived for recording all the relation strengths. More accurately, the weights of edges in the adjoint graph is determined by the corresponding RSM. 
			\begin{definition}[Adjoint complete digraph]
			Suppose $G$ is some directed graph and $g$ is some RSM. Let $E = V_G \times V_G$ be a new set of edges whose weights are assigned by $g$. Then the adjoint complete digraph $adj(G,g):=\{V_G, E\}$. 
			\end{definition}

			The definition of communities uses CP to give a threshold of the relation strength. That means, if the relation strength is not strong enough, the relation is ignored during the community detection. Moreover, for any pair of nodes, the definition of communities requires the enough strengths of the relations in both two directions. Hence, we can remove the relations unsatisfying the requirement to simplify our graph without changing the result of community detection. With this trick in mind, we have the following transformation. 

			\begin{definition}[Refinement transformation]
			Suppose $G$ is a directed weighted graph, $D_G \subseteq E_G$, and $\epsilon \in \mathbb{R}^*$ is some CP. The refinement transformation is defined like this.
				$$R_{\epsilon}(D_G):= \{(u,v)\in D_G: |(u,v)| \leq \epsilon \wedge |(v,u)| \leq \epsilon \}$$ 
			Besides, all the edges' weights are set to 1 after applying the transformation.
			\end{definition}

			\begin{remark}
			For convenience, the relations $(u,v)$ and $(v,u)$ are together denoted $u \leftrightarrow v$. In this case the weight is not applicable.
			\end{remark}

			The definition of refinement transformation shows that if some edges $(u,v)$ is in $R_{\epsilon}(D_G)$, then so is $(v,u)$. Moreover, the weights become unnecessary since they all equal one. Therefore, in $R_{\epsilon}(D_G)$, there is no need for us to consider the directions and weights of the edges anymore. So, for now on, $R_{\epsilon}(D_G)$ is thought of a set of undirected unweighted edges. Moreover, if $(u,v) \in R_{\epsilon}(D_G)$, then we say the relation between $u$ and $v$ is reserved. Or briefly, $u \leftrightarrow v$ is reserved.

			In fact, the refinement transformation is a higher order function that applies a Boolean function to each relation in the set of edges. The Boolean function here determines whether the given relation is strong enough to be considered in the community detection. So for a certain refinement transformation, the reservation of the relation depends on the strength of the relation itself rather than the topology in which the relation is.
	
			\begin{lemma}\label{refinementPreserve}
			Suppose $G$ is some directed graph.	If $F_1 \subseteq F_2 \subseteq E_G$, then $R_{\epsilon}(F_1) \subseteq R_{\epsilon}(F_2)$. 
			\end{lemma}
			\begin{proof}
			Pick $(u,v) \in R_{\epsilon}(F_1)$ arbitrary. So $u\leftrightarrow v$ is in $F_1$ and reserved after applying the refinement transformation. Since $F_1 \subseteq F_2$, then $u\leftrightarrow v \in F_2$. So $(u,v) \in R_{\epsilon}(F_2)$.
			\end{proof}

			Then the adjoint complete digraph can derive a simplified undirected unweighted graph whose edges represent the two-direction relations strong enough to construct communities.
	
			\begin{definition}[Effective edge graph]
			Suppose $G$ is some directed graph. $g$ is some RSM. $\epsilon$ is some CP. Then the effective edge graph is $(V_G, R_{\epsilon}(E_{adj(G,g)}))$ and denoted by $eeg_{(g,\epsilon)}(G)$. Moreover, suppose the vertices set $A^G$ is a subset of $V_G$. The full subgraphs of $eeg_{(g,\epsilon)}(G)$ over $A^G$ is denoted by $eeg_{(g,\epsilon)}(G)[A^G]$.
			\end{definition}
  
			\begin{lemma}\label{3rd definition of communities}
			Let $g$ be some RSM, $\epsilon \in \mathbb{R^*}$ some CP and $G$ some directed weighted graph. Assume $W^G \subseteq V_G$. Then the vertices set $W^G$ is a community if and only if  for all $u,v \in W^G$, $u \leftrightarrow v$ is reserved after applying the refinement transformation. 
			\end{lemma}

			\begin{figure*}[!t]
			\centering
			\includegraphics[width=0.8\textwidth]{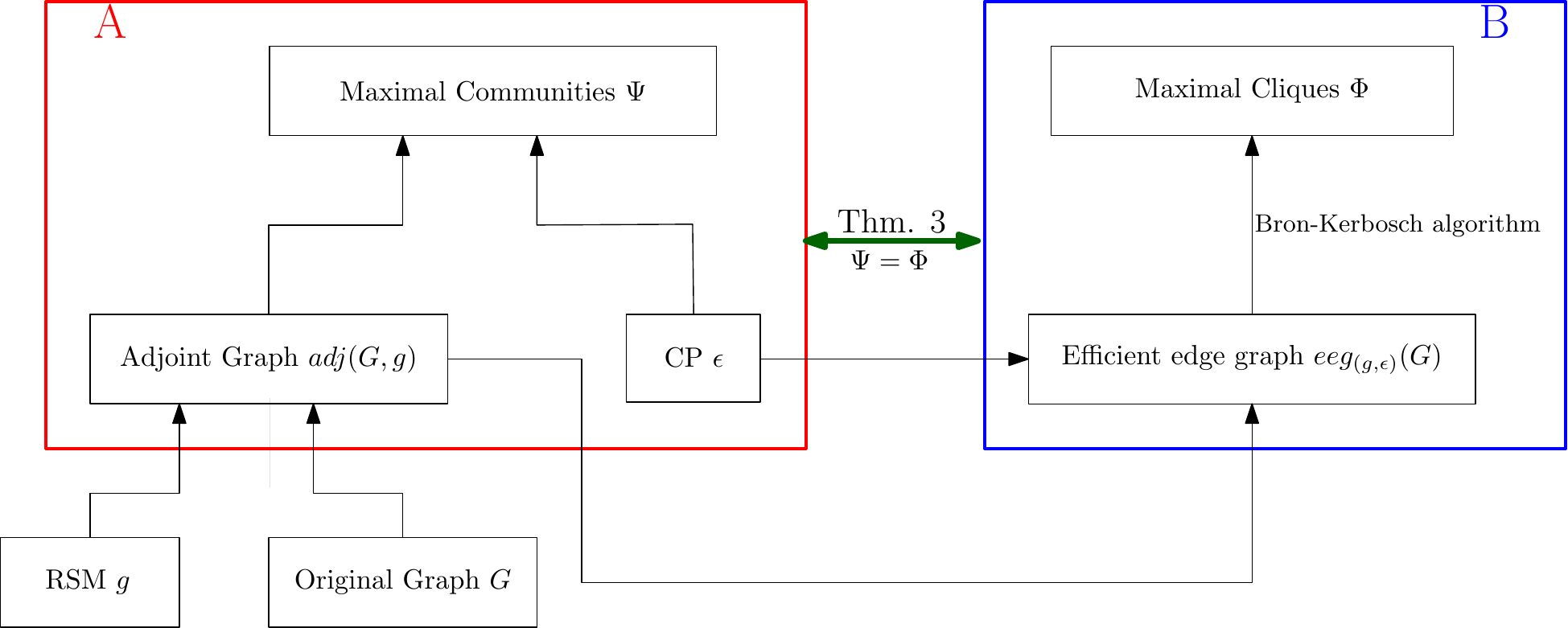}
				\caption{The Concept Graph of The Algorithm}\label{propStructure}
			\end{figure*}

			\begin{proof}
				The refinement transformation will remove all the relations that cannot be used in a community structure. In other words, if all relations are reserved after applying refinement transformation, the relation in any pair of nodes is strong enough. This is exactly what the definition of communities requires. So, $W^G$ is a community. On the other hand, if $W^G$ is a community, the relation (in both directions) in any pair of nodes should be strong enough. Thus, all of them are reserved after applying the refinement transformation.
			\end{proof}
		
			\begin{lemma}
				Any full subgraph of a complete graph is again complete.
			\end{lemma}
			\begin{proof}
				Suppose $S$ is a full subgraph of $C_n$ for some $n$. Then for arbitrary vertices $u$ and $v$ in $S$, edge $(u,v) \in E_{C_n}$. So by the definition of full subgraphs, $(u,v) \in E_{S}$. That is, there is an edge in an arbitrary pair of nodes in $S$. So $S$ is a complete graph.
			\end{proof}
				
			\begin{theorem}\label{4th definition of communities}
				Suppose $G$ is some directed graph, $g$ some RSM and $\epsilon$ some CP. $A^G \subseteq V_G$ is a community if and only if $eeg_{(g,\epsilon)}(G)[A^G]$ is complete. 
			\end{theorem}
			
			\begin{proof}
			
			\
			
			\noindent
			$(\Rightarrow)$
			Suppose $A^G$ is a community. Pick vertices $u$ and $v$ in $A^G$ arbitrary. Since $A^G$ is a community, then all the relations will be reserved after applying the refinement transformation. Moreover, since $adj(G,g)[A^G]$ is complete, then $u \leftrightarrow v \in E_{adj(G,g)[A^G]}$. Thus, $(u,v) \in  R_{\epsilon}(E_{adj(G,g)[A^G]})$. Since we pick $u$ and $v$ arbitrary in $A^G$, there is an edge between any pair of nodes in $eeg_{(g,\epsilon)}(G)[A^G]$. Hence, $eeg_{(g,\epsilon)}(G)[A^G]$ is complete.
				
			\noindent
			$(\Leftarrow)$
			Suppose $eeg_{(g,\epsilon)}(G)[A^G]$ is complete. Then $\forall (u,v) \in A^G \times A^G$, $(u,v) \in R_{\epsilon}(E_{adj(G,g)[A^G]})$. Since $eeg(G,g)[A^G]$ is complete, all the relations in $E_{adj(G,g)[A^G]}$ are reserved after applying the refinement function. Therefore, $A^G$ is a community.	
			\end{proof}

			It is easy to find that all single nodes can be considered as a community because they relate to themselves trivially, and RSM is zero. However, this kind of result does not follow our intuition since the community should be some set of nodes. The definition of the maximal community tackles this problem. For a better understanding of the definition, a theorem needs to be introduced first.

			\begin{theorem}\label{InnerCommunity}
				 Suppose $G$ is a directed graph, $\epsilon$ a CP and $g$ an RSM. Besides, $A^G \subseteq B^G \subseteq V_G$. If $B^G$ is a community, then so is $A^G$ 
			\end{theorem}
			
			\begin{proof}
				Since $B^G$ is a community, then by Theorem \ref{4th definition of communities}, $eeg_{(g,\epsilon)}(G)[B^G]$ is complete. Besides, since $A^G \subseteq B^G$, then by Lemma~\ref{refinementPreserve}, $eeg_{(g,\epsilon)}(G)[A^G]$ is a subgraph of $eeg_{(g,\epsilon)}(G)[B^G]$. 
				
				Moreover, pick $u,v \in A^G$ arbitrary. Then $u,v \in B^G$ as well. Notice that $eeg_{(g,\epsilon)}(G)[B^G]$ is complete. So the relation $u \leftrightarrow v \in E_{adj(G,g)[B^G]}$ is reserved, which implies $(u,v)\in E_{eeg_{(g,\epsilon)}(G)[A^G]}$. Thus, $eeg_{(g,\epsilon)}(G)[A^G]$ is complete. So $A^G$ is a community as well.
			\end{proof}
			
			Theorem \ref{InnerCommunity} shows that, if $B_G$ can be considered as a community with respect to some RSM and CP, then all the subsets of $B_G$ can be considered as a community. This observation leads to the definition of maximal community.
			
			\begin{definition}[Maximal community]
				Suppose $G$ is a directed graph. RSM and CP are given. Moreover, $A^G$ is a subset of $V_G$. Then $A^G$ is a maximal community if and only if
				\begin{itemize}
					\item[1.] $A^G$ is a community, and
					\item[2.] There is no $B^G \subseteq V_G$ such that $B^G$ is a community and $A^G\subsetneq B^G$.
				\end{itemize}
			\end{definition}

			With the maximal community definition in mind, we introduce an algorithm to detect them if RSM and CP are specified. For easier explanation, we define \textbf{problem A} as this:
			
			\begin{definition}[Problem A]
				Given some adjoint graph $adj(G,g)$ and CP $\epsilon$, find all the maximal communities in $G$ (the set of the maximal communities is denoted $\Psi$).
			\end{definition}
			
			\begin{definition}[Problem B]
				Given the effective edge graph $eeg_{(g,\epsilon)}(G)$, find the all the maximal cliques in $eeg_{(g,\epsilon)}(G)$ (the set of the maximal cliques is denoted by $\Phi$).
			\end{definition}

			The following theorem shows the equivalence of problem A and problem B. 
						
			\begin{theorem}\label{maxCommunityAndMaxClique}
				Suppose $G$ is a directed graph, $\epsilon$ some CP, $g$ an RSM and $A^G$ a subset of $V_G$. Then $A^G$ is a maximal community if and only if
				$eef_{(g,\epsilon )}(G)[A^G]$ is a maximal clique in graph $eef_{(g,\epsilon)}(G)$. Therefore, $\Psi = \Phi$.
			\end{theorem}
			
			\begin{proof}
			
			\
			
			\noindent
				$(\Rightarrow)$ 
				Suppose $A^G$ is a maximal community. Then since $A_G$ is a community, then $eef_{(g,\epsilon )}(G)[A^G]$ is complete. So it is a clique. Assume $eef_{(g,\epsilon )}(G)[A^G]$ is not maximal. Then there exists some vertices set $B^G$ such that $eef_{(g,\epsilon )}(G)[A^G] \subsetneq eef_{(g,\epsilon )}(G)[B^G]$ and $ eef_{(g,\epsilon )}(G)[B^G]$ is complete. So $A^G \subseteq B^G$. Moreover, since both two graphs are complete, the equality cannot hold. Otherwise, $eef_{(g,\epsilon )}(G)[A^G] = eef_{(g,\epsilon )}(G)[B^G]$. So we have $A^G \subsetneq B^G$. Besides, since $eef_{(g,\epsilon )}(G)[B^G]$ is complete, $B^G$ is a community. Hence, $A^G$ cannot be a maximal community, which is a contradiction.
					
			\noindent
			$(\Leftarrow)$ 
			Suppose $eef_{(g,\epsilon )}(G)[A^G]$ is a maximal clique. Since $eef_{(g,\epsilon )}(G)[A^G]$ is complete, then $A^G$ is a community. Assume $A^G$ is not maximal, then there exists some community $B^G$ such that $A^G \subsetneq B^G$. Then $eef_{(g,\epsilon )}(G)[B^G]$ is complete. Moreover, we have $eef_{(g,\epsilon )}(G)[A^G] \subsetneq eef_{(g,\epsilon )}(G)[A^G]$ by Lemma~\ref{refinementPreserve}. Therefore, $eef_{(g,\epsilon )}(G)[A^G]$ cannot be maximal, which is a contradiction. 
			\end{proof}

			\begin{remark}
				Bron-Kerbosch algorithm \citep{findMaximalClique} is a well-known algorithm to find maximal cliques. Thence, by Theorem \ref{maxCommunityAndMaxClique}, we can reduce problem A to problem B and apply Bron-Kerbosch algorithm to find $\Phi$, which equals $\Psi$.
			\end{remark}

			Figure~\ref{propStructure} shows the relationships among the important concepts and transformations introduced. In more details, suppose $G$ is some graph, $g$ some RSM and $\epsilon$ some CP. Moreover, all the vertices in $G$ have been indexed from $1$ to $|V_G|$. Then we have the following algorithm,
			
			\begin{algorithm}
							\begin{algorithmic}[1]
								\Procedure{findMaxiamlCommunities}{}
									\For{$(i,j) \in V_G \times V_G$ where $i \leq j$}
											\If{
												$g(i,j)\leq \epsilon \And g(j,i)\leq \epsilon$
											}
												\State $EEG[i,j] = 1$
											\Else
												\State $EEG[i,j] = 0$
											\EndIf
									\EndFor
									\State $\textit{SetOfMaximalCommunities}\hspace{0.5em}\Phi\hspace{0.5em}\gets$ 
										\\ \hfill $\text{Bron-Kerbosch algorithm(EEG)}$
								\EndProcedure
							\end{algorithmic}
			\end{algorithm}
			
	\section{Demonstration}\label{secDemo}
		In this section, we demonstrate how our new algorithm works by applying it on Zachary's karate club network \citep{Zac77}. We choose resistance distance \citep{orgResistanceDistance} as our RSM. 
		
		\subsection{The current model}
		The definition of communities indicates that some certain RSM is required. We have shown that SDF is RSMFTN in Example~\ref{SDFISRSM}, so that SDF is RSM. Although many community detection algorithms work on SDF, it may not always give a reasonable result. Intuitively, the relation of a pair nodes will get stronger if there are more paths between them. However, SDF does not consider this case (see Figure~\ref{SDFcounterExample}). More specifically, in SDF view, the relation will not get stronger unless a path shorter than the previous shortest path is added.
	
	\begin{figure}[h]
		\centering
		\begin{subfigure}[b]{0.4\columnwidth}
			\includegraphics[width=\columnwidth]{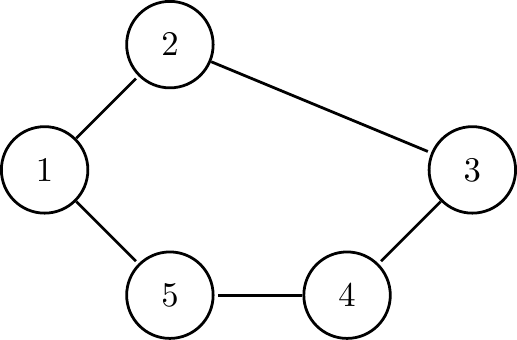}
			\caption{}
		\end{subfigure}\hspace{1em}
		\begin{subfigure}[b]{0.4\columnwidth}
			\includegraphics[width=\columnwidth]{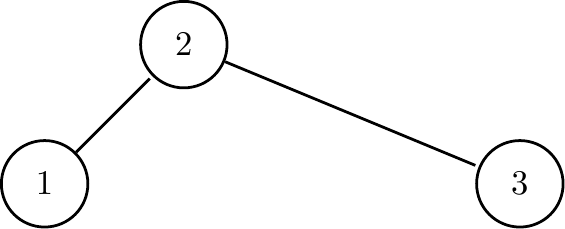}
			\caption{}
		\end{subfigure}
		\caption{Both in (a) and (b), the shortest distance between Node1 and Node3 is 2. So in SDF view, the relation strengths between Node1 and Node3 in these two cases are identical. However, there is one more path between Node1 and Node3 in (a). So, intuitively, the relation between Node1 and Node3 in (a) should be stronger than the one in (b).}\label{SDFcounterExample}
	\end{figure}
		 	
		In order to avoid this problem, we try to use the Klein and Randic's effective resistance function (ERF) \citep{orgResistanceDistance} to measure the relation strength instead of SDF. 
		
		\subsection{Klein and Randic's effective resistance}
		Suppose $G$ is an indirectly connected graph. Then $G$ can be considered as an electrical network that all the edges are resistances with the corresponding weight values (if $G$ is an unweighted graph, then the resistances of all edges are one).
		
		Let $u$ and $v$ be two vertices in the graph. Then the effective resistance of these two vertices can be defined like this:
		\begin{definition}
			Let the voltage of $u$ be $U$ and the one of $v$ be $0$. We can measure the current $I$ from $u$ to $v$. Then the efficient resistance $R(u,v)$ between $u$ and $v$ is $\frac{U}{I}$. Briefly, $R(u,v)=\frac{U}{I}$.
		\end{definition}

		\subsubsection{Algorithm to get efficient resistance distance}
		Klein and Randic \citep{orgResistanceDistance} also provide an algorithm to compute the resistance distance for a connected indirect weighted graph.
		
		Suppose graph $G$ is connected. Let $A$ be the adjacent matrix and $D$ the diagonal degree matrix of $G$. It is worthy to note that, in a weighted indirected graph, the degree of a vertex is the sum of the weights of all its adjacent edges. Then the Laplacian matrix $L$ can be computed using formula $L = D - A$. Let $L^{\dag}$ be the generalized inverse \citep{generalizedInverse} of $L$. 
		Then the efficient resistance distance $R_{i,j}$ of any pair of vertices $(i,j)$ in graph $G$ can be obtained by $$R_{i,j} = L^{\dag}_{i,i}+L^{\dag}_{j,j}-2L^{\dag}_{i,j}$$
		And we usually call the corresponding matrix $R$ \textbf{resistance matrix}.

		\subsubsection{ERF is an RSMFTN}
		Since the definition of community is based on RSM, we have to prove ERF is an RSM first. Essentially, in this case, the relations among the nodes are derived from the electron flow in the wires among the vertices. So we need to consider the criteria of RSMFSN.		
		
		\begin{lemma}\label{resistance is distance}
			Resistance is distance. That is, the resistance satisfies the following properties:
			\begin{enumerate}
				\item $R_{a,b} \geq 0$
				\item $R_{a,b} = 0 \Leftrightarrow a = b$
				\item $R_{a,b} = R_{b,a}$
				\item $R_{a,c} + R_{c,b} \geq R_{a,b}$
			\end{enumerate}
		\end{lemma}
		
		\begin{lemma}\label{RDFCutNode}
			Let $x$ be a cut-point of a connected graph, and let $a$ and $b$ be points occurring in different components which arise upon deletion of $x$. Then, $$R_{a,b} = R_{a,x} + R_{x,b}$$
		\end{lemma}
		
		\begin{remark}
			The proofs of Lemma~\ref{resistance is distance} and Lemma~\ref{RDFCutNode} have been given by Klein and Randic \citep{orgResistanceDistance}.
		\end{remark}

		\begin{lemma}\label{RDF Property5}
			RDF satisfies the property 5 of RSM.
		\end{lemma}
		
		\begin{proof}
			Suppose $G$ is some graph and $G'$ is same as $G$ but the edges weights in $G'$ are all $\alpha$ times greater than the ones in $G$.
			Let $A$ and $A'$ be the adjacent matrixes of $G$ and $G'$ respectively. Then we have $A' = \alpha A$. So for the corresponding degree matrixes $D$ and $D'$, we also have $D' =  \alpha D$. Therefore, we have 
			$$L' = D' - A' = \alpha D - \alpha A = \alpha (D-A) = \alpha L$$
			
			Let $L^{\dag}$ and $L'^{\dag}$ be the generalized inverse of $L$ and $L'$ respectively. Then by the definition of the generalized inverse, we have 
			\begin{equation}\label{orgLpInv}
				L L^{\dag} L = L
			\end{equation}
			\begin{equation}\label{primeLpInv}
				L' L'^{\dag} L' = L'
			\end{equation}
			Since $L' = \alpha L$, we can simplify equation \ref{primeLpInv}
			\begin{align*}
				L' L'^{\dag} L' & = L'  &(\Leftrightarrow)\\
				(\alpha L) L'^{\dag} (\alpha L') & = (\alpha L') &(\Leftrightarrow)\\
				L (\alpha L'^{\dag})  L & = L	&
			\end{align*}
			Comparing with equation \ref{orgLpInv}, $\alpha L'^{\dag}$ has the same function as $L^{\dag}$. Since the final result does not rely on the choice of the generalized inverse matrix, we can let $L^{\dag}$ be the one satisfying the equation 
			\begin{equation}
				\alpha  L^{\dag} = L'^{\dag}
			\end{equation}
			Hence, we have
			\begin{align*}
				R'_{i,j} &= L'^{\dag}_{i,i}+L'^{\dag}_{j,j}-2L'^{\dag}_{i,j}\\
				& = \alpha  L^{\dag} + \alpha  L^{\dag} - 2 \alpha  L^{\dag}\\
				& = \alpha (L^{\dag}_{i,i}+L^{\dag}_{j,j}-2L^{\dag}_{i,j})\\
				& = \alpha R_{i,j}
			\end{align*}
		\end{proof}
		
		\begin{proposition}
			ERF is an RSMFSN.
		\end{proposition}
		\begin{proof}
			We can define that the resistance distance of a pair of vertices is infinite if there is no path between them. Then the proposition is immediate from Lemmas~\ref{resistance is distance}-\ref{RDF Property5}.
		\end{proof}
		
		Therefore, ERF is RSM.
	
		\subsection{Community detection in Zachary's karate club}
		\begin{figure}[h]
			\centering
			\includegraphics{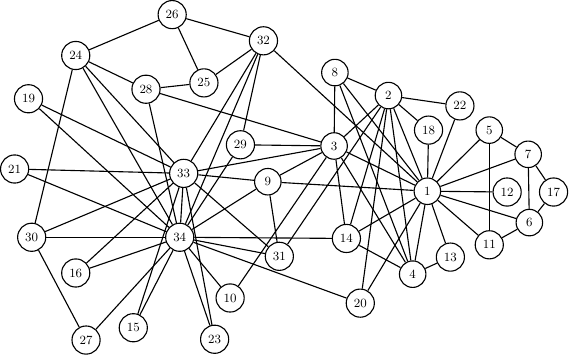}
			\caption{Zachary's karate club}\label{ZKC}
		\end{figure}
		The graph we use for demonstration is Zachary's karate club (Figure~\ref{ZKC}) \citep{Zac77}, which is a popular test case in community detection research.

		At first, we need to choose some proper CP, which is the lower bound of the relation strength within the communities. Here, we let $CP = 1.5$.
	
		Then, we use Klein-Randic method to compute the resistance distance for each pair of nodes in the network and get the corresponding resistance matrix $R$.
		
		After that, we get the corresponding adjoint graph (adj) from $R$ and remove all the edges whose weights are greater than CP. So we get the efficient edges graph (eeg). 
		
		Then we apply Bron-Kerbosch algorithm on eeg and get a list of maximal communities.
		
		In Figure~\ref{Maximal Communities}, we plot those maximal communities in the original graph. Here, we have three maximal communities represented by red, blue and yellow respectively. Some nodes are multi-colour, which means they belong to various maximal communities simultaneously.
	
	\begin{figure*}[h]
		\centering
		\includegraphics[width=0.7\textwidth]{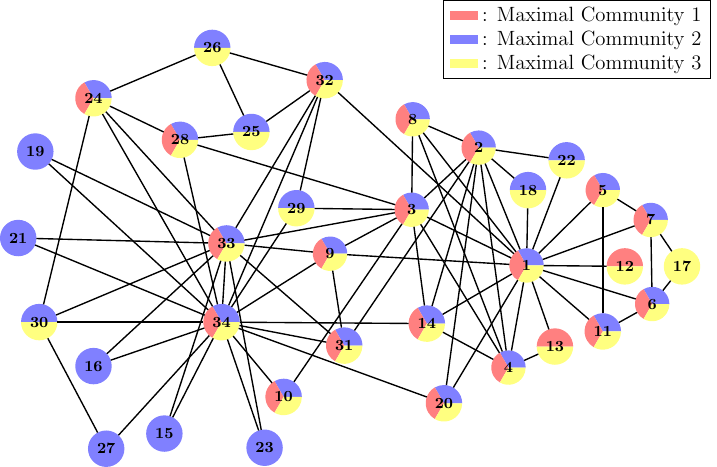}
		\caption{Maximal Communities ($CP = 1.5$)}\label{Maximal Communities}
	\end{figure*}
			
	\section{Conclusion}\label{secConclusion}
	In this paper, we discussed two most common types of networks: transmission networks and similarity networks. Two corresponding relation strength measurements (RSMFTN and RSMFSN) are defined. And we reduce them into a unified graph model. Based on this, we provide a general definition of communities and derive a corresponding detection algorithm. At last, we give a demonstration to show how the algorithm works.
	
	Our paper gives a general procedure to detect community structures in practical networks. Readers can specialize our algorithm to derive theirs according to the problems confronting them. 
	
	\section{Limitations and Future Work}\label{secLimit}
	Generally, RSMs consider the whole network topology. So does the algorithm to find the maximal community structures. While the algorithm gives the accurate results, it is NP hard. So our algorithm may not suit the community detections in dynamic networks. Besides, the definition we give in this paper is based on the absolute strengths among the nodes. So the users should always give a proper community parameter $\epsilon$, which is hard to find sometimes.
	
	Although we proved that SDF and ERF are RSMFTN, many other RSMFTNs still wait to be discovered. Besides, the definition based on the absolute relation strength should derive a corresponding definition based on the relative relation strength. The key point is how to give a general definition of the neighbour nodes when applying different RSMs. 
	


%

\bibliographystyle{IEEEtran}
\bibliography{bibliography}

\begin{IEEEbiography}
[{\includegraphics[width=1in,height=1.25in,clip,keepaspectratio]{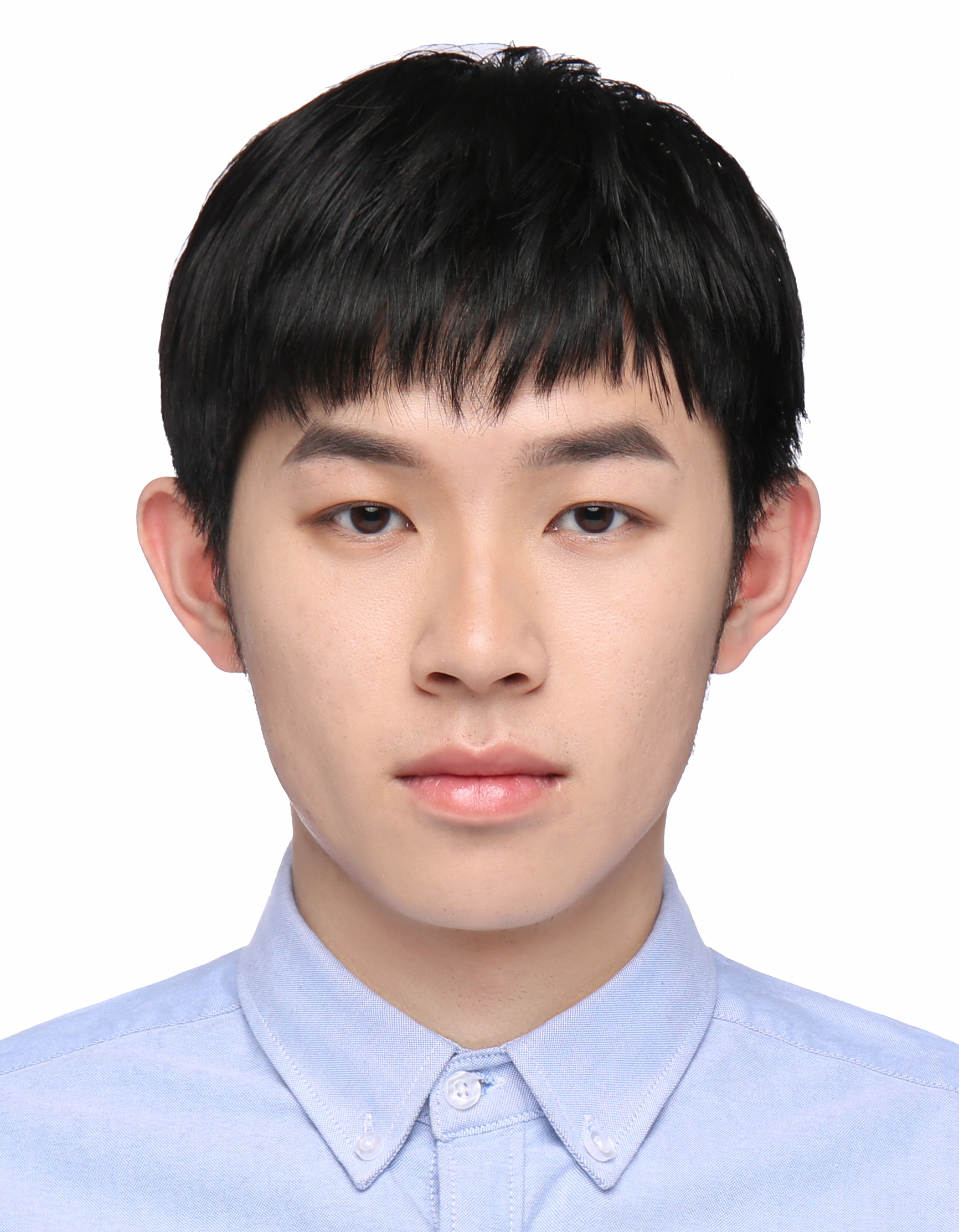}}]{Haoye Lu} received his Bachelor of Science Joint degree Honours in Computer Science and Mathematics from University of Ottawa, Canada, in 2017. He was the recipient of the University Silver Medal for his second highest academic achievement in Faculty of Science. 

Haoye joined the School of Electrical Engineering and Computer Science (EECS), University of Ottawa in 2017 and is currently a master student pursuing Master of Science degree in Computer Science. He was also the recipient of the Full International Scholarship. He is a reviewer of 2017 IEEE Global Communications Conference and has published papers in quantum communication and artificial intelligence fields. His research interests include quantum communication, artificial intelligence and networks structures. 
\end{IEEEbiography}
\begin{IEEEbiography}
[{\includegraphics[width=1in,height=1.25in,clip,keepaspectratio]{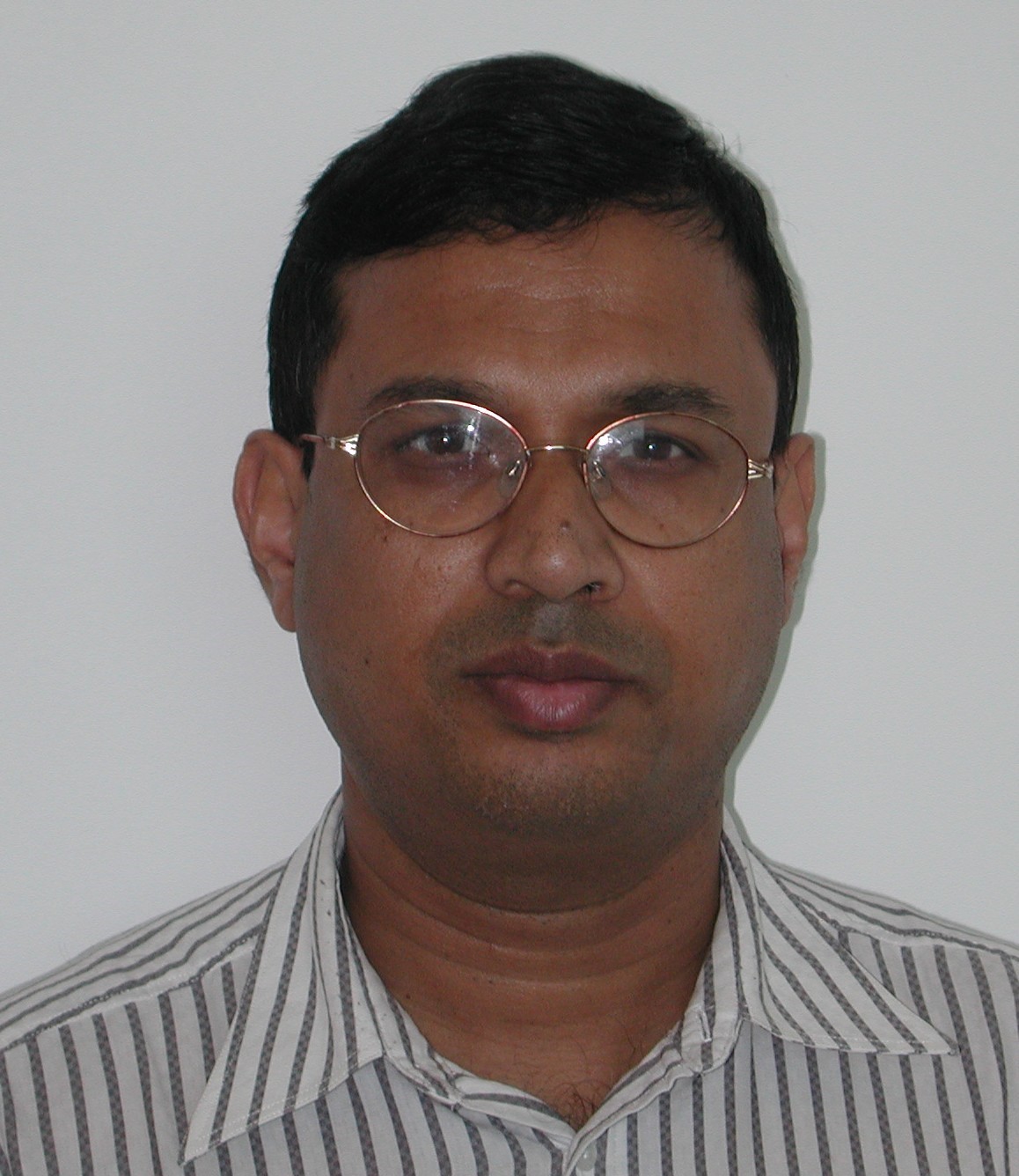}}]{Amiya Nayak} received his B.Math. degree in Computer Science and Combinatorics and Optimization from University of Waterloo, Canada, in 1981, and Ph.D. in Systems and Computer Engineering from Carleton University, Canada, in 1991.
He has over 17 years of industrial experience in software engineering, avionics and navigation systems, simulation and system level performance analysis.
He has served in the Editorial Board of several journals, including IEEE Transactions on Parallel \& Distributed Systems, International Journal of Parallel, Emergent and Distributed Systems, Journal of Sensor and Actuator Networks, and EURASIP Journal of Wireless Communications and Networking. Currently, he is a Full Professor at the School of Electrical Engineering and Computer Science at the University of Ottawa. His research interests include software-defined networking, mobile computing, wireless sensor networks, and vehicular ad hoc networks.
\end{IEEEbiography}

\end{document}